\documentclass[11pt]{article}

\usepackage{amsthm,amsmath,amssymb}
\usepackage[bookmarks=true,hypertexnames=false,pagebackref]{hyperref}
\hypersetup{colorlinks=true, citecolor=blue, linkcolor=red,
  urlcolor=blue}
\usepackage{newpxtext}
\usepackage[libertine,vvarbb]{newtxmath}

\usepackage{fullpage}
\usepackage{hyperref}
\usepackage{cleveref}
\usepackage[utf8]{inputenc}
\usepackage{framed}
\usepackage{enumerate}
\usepackage{graphicx}
\usepackage{float} 
\usepackage{subfigure}

\usepackage[ruled,vlined]{algorithm2e}
\usepackage{algorithmic}

\usepackage{tikz}
\usepackage{complexity}
\usepackage{ulem}

\newtheorem{theorem}{Theorem}[section]

\newtheorem{fact}[theorem]{Fact}
\newtheorem{corollary}[theorem]{Corollary}
\newtheorem{lemma}[theorem]{Lemma}
\newtheorem{definition}[theorem]{Definition}
\newtheorem{conjecture}[theorem]{Conjecture}

\newcommand{\OCC}{\mathrm{OCC}}
\newcommand{\NCC}{\mathrm{OCC^{NOF}}}

\newcommand{\ind}{\mathrm{Ind}}

\newcommand{\gip}{\mathrm{GIP}}

\newcommand{\ext}{\mathrm{ext}}

\newcommand{\EQUAL}{\mathrm{EQ}}

\title{Deterministic Lifting Theorems for One-Way Number-on-Forehead Communication}

\begin{document}
\author{
Guangxu Yang \thanks{Research supported by NSF CAREER award 2141536.\\ Thomas Lord Department of Computer Science, University of Southern California.  \\ Email: \texttt{\{guangxuy, jiapengz\}@usc.edu}}
\\
\and
Jiapeng Zhang \footnotemark[1]
}
\maketitle
\begin{abstract}

Lifting theorems are one of the most powerful tools for proving communication lower bounds, with numerous downstream applications in proof complexity, monotone circuit lower bounds, data structures, and combinatorial optimization. However, to the best of our knowledge, prior lifting theorems have primarily focused on the two-party communication.

In this paper, we propose a new lifting theorem that establishes connections between the \textit{two-party} communication and the \textit{Number-on-Forehead (NOF)} communication model. Specifically, we present a deterministic lifting theorem that translates one-way two-party communication lower bounds into one-way NOF lower bounds.

Our lifting theorem yields two applications. First, we obtain an optimal explicit separation between randomized and deterministic \textit{one-way NOF communication}, even in the multi-player setting. This improves the prior square-root vs. constant separation for three players established by Kelley and Lyu (arXiv 2025). Second, we achieve \textit{optimal separations} between one-round and two-round deterministic NOF communication, improving upon the previous separation of \( \Omega\left(\frac{n^{1/(k-1)}}{k^k}\right) \) vs. \( O(\log n) \) for \( k \) players, as shown by Viola and Wigderson (FOCS 2007).

On the technical side, we use the \textit{generalized inner product function over large fields (GIP)} as the gadget. As a bonus result, we construct an explicit function that for any number of players $k$, at least $\Omega(n / 2^k)$ bits of communication are required to solve it, even for randomized NOF protocols. This improves the previous best bound of $\Omega(n / 4^k)$ for generalized inner product over the binary field, established by Babai, Nisan, and Szegedy (STOC 1989). 

Beyond the lifting theorems, we also apply our techniques to the disjointness problem. In particular, we provide a new proof that the deterministic one-way three-party NOF communication complexity of set disjointness is $\Omega(n)$, further demonstrating the broader applicability of our methods.

\end{abstract}
\section{Introduction}
Lifting theorems are a generic method for translating lower bounds from weaker computational models to relatively stronger ones. A representative example of lifting theorems is the query-to-communication lifting theorems, which convert lower bounds in query complexity into communication complexity lower bounds, using a suitable base function composed with a gadget.

There has been a long line of work developing query-to-communication lifting theorems, with diverse applications in various areas such as monotone circuit complexity, proof complexity, combinatorial optimization, and more \cite{raz1997separation, zhang2009tightness, goos2018deterministic,pitassi2017strongly,goos2020query,chattopadhyay2019query,lovett2020lifting,Collision,mao2025gadgetless}. Though these lifting theorems have achieved significant success, to the best of our knowledge, all existing lifting theorems have focused on two-party communication models (or on a slight generalization known as the Number-in-Hand model). No lifting theorems have yet been established for the more powerful models, such as the Number-on-Forehead (NOF) models. 

Compared to two-party communication lower bounds, lower bounds in NOF models are known to have more connections to several major open problems. However, existing lower bound techniques for NOF models remain quite limited. Motivated by this, papers by Kumar, Meka, and Zuckerman~\cite{kumar2020bounded} and Toniann Pitassi~\cite{NOFlifting} raised a natural question: 

\begin{center}
\textit{Can we develop lifting theorems in NOF models and derive interesting applications?}
\end{center}

Though this problem is well-motivated, it remains wide open. Furthermore, we do not even have good candidates for conjectured statements for the lifting theorem in NOF models. In this paper, we propose a lifting theorem that converts two-party communication lower bounds into lower bounds in the NOF models. As a first step in this direction, we prove a deterministic lifting theorem that lifts one-way communication lower bounds in the two-party setting to one-way communication lower bounds in NOF protocols. Although the one-way NOF model may seem like a restricted variant of the general NOF model, since each player is allowed to speak only once in a fixed order, it still enables many surprising applications. Notably, it has been instrumental in proving strong lower bounds in areas such as circuit complexity~\cite{HG90, pudlak1997boolean, chakrabarti2007lower, VW07}, cryptography~\cite{chor1998private, brody2017position}, and streaming algorithms~\cite{kallaugher2019complexity, VW07}. Despite its restricted nature, proving lower bounds in the one-way NOF model remains a highly challenging and fundamental problem.

\subsection{Our Contribution}

In this paper, we propose a new form of lifting two-party communication problems into NOF communication problems.

\begin{definition}
For any two-party function $f: [q] \times [q] \rightarrow \{0,1\}$ and a gadget function $g^{k}: [N]^{k} \rightarrow [q]$, the lifted problem over $(k+1)$ players, denoted by, $
f \circ g^{k}: [N] \times \cdots \times [N] \times [q] \rightarrow \{0,1\}$,
is defined as
\[
f \circ g^{k}(x_1, \dots, x_k, z) = f(z, g^{k}(x_1, \dots, x_k)).
\]
In the NOF setting, we assume that the $(k+1)$-th player knows the input $(x_1,\dots,x_k)$, while each of the remaining $i$-th players knows all inputs except $x_i$.
\end{definition}
We observe that many well-known problems in the NOF model can be expressed as lifted problems of the form $f \circ g$. Examples include \textit{Generalized Inner Product}~\cite{rao2020communication}, \textit{Set Disjointness}~\cite{rao2020communication}, \textit{Pointer Jumping}~\cite{VW07,chakrabarti2007lower}, \textit{Shifting}~\cite{pudlak1997boolean}, and the \textit{Exactly-$N$} problem~\cite{chandra1983multi}. As a quick example, consider the \textit{Exactly-$N$} problem: Alice, Bob, and Charlie are given integers $x$, $y$, and $z$, and the goal is to determine whether $x + y + z = N$. We can express this as a composition $\EQUAL \circ g$ by defining $g(x, y):= N - x - y$, and $\EQUAL(z, s) = 1$ if and only if $z = s$. Then, the condition $x + y + z = N$ holds if and only if $f(z, g(x, y)) = 1$.

In this paper, we primarily use the generalized inner production function as the gadget $g$.

\begin{definition}
Let $q$ be a prime and $k, r > 0$. We define the function $\gip_{q,r}^{k}: (\mathbb{F}_{q}^{r})^{k} \rightarrow \mathbb{F}_q$ by
\[
\gip_{q,r}^{k}(x_1,\dots,x_k) = \sum_{j \in [r]} \prod_{i \in [k]} x_{i,j},
\]
where $\mathbb{F}_q$ is a prime field, and all arithmetic operations are over $\mathbb{F}_q$. When $q, r, k$ are clear from context, we write $\gip(x_1,\dots,x_k)$ for simplicity.
\end{definition}

Using the generalized inner product function, we show that for any two-party function \( f \), the \textit{one-way communication} lower bound of \( f \), denoted \( \mathrm{OCC}(f) \), can be lifted to a \textit{one-way NOF} lower bound of \( f \circ \gip \), denoted \( \mathrm{OCC^{NOF}}(f \circ \gip) \).

\begin{theorem}\label{main_theorem}
Let $q$ be a prime and $k, r \geq 2^{k+1}$. For any Boolean function \( f : [q] \times [q] \rightarrow \{0,1\} \), we have
\[
\mathrm{OCC^{NOF}}(f \circ \gip_{q,r}^{k}) = \Theta(\mathrm{OCC}(f)).
\]
\end{theorem}

Although we only prove the theorem for one-way communication, we conjecture that it also holds for general NOF communication. 

\subsubsection{Applications of Our Main Theorem}
We now discuss some applications of our main theorem.
\paragraph{Deterministic vs. randomized separation.}
One of the central goals in communication complexity is understanding the power of randomized communication. Separations between deterministic and randomized communication complexity in NOF have received a good amount of attention in recent years \cite{beame2010separating, kelley2024explicit, jaber2025, kelley2025efficientsiftinggridnorms}. Beame, David, Pitassi, and Woelfel~\cite{beame2010separating} showed the existence of a three-party function \( f : [N]^3 \to \{0,1\} \) with randomized NOF complexity \( O(1) \) and deterministic NOF complexity \( \Omega(\log N) \). However, their result was non-explicit, relying on a counting argument.

For explicit constructions, the \textit{Exactly-N} function introduced by \cite{chandra1983multi} has long been considered a strong candidate for such separations. While its randomized three-party complexity is \( O(1) \), a deterministic lower bound was not known until the recent breakthrough by Jaber, Liu, Lovett, Ostuni, and Sawhney~\cite{jaber2025}, who showed an $\Omega((\log N)^{\Omega(1)})$ lower bound. For other explicit constructions, Kelley and Lyu~\cite{kelley2025efficientsiftinggridnorms} established an $\Omega\left(\sqrt{\log N}\right)$ vs.\ $O(1)$ separation.

Though these advances are very impressive, two key challenges remain. First, for existing approaches including~\cite{kelley2024explicit, jaber2025, kelley2025efficientsiftinggridnorms}, it is still unknown how to break the square-root barrier. That is, it remains unclear how to achieve an \( \Omega(\log N) \) vs. \( O(1) \)  optimal separation, even in the one-way setting. Second, proving deterministic lower bounds for more than three players remains difficult. For instance, the result of~\cite{jaber2025} yields only an \( \Omega(\log\log\log\log N) \) vs.\ \( O(1) \) separation for four players.

In this paper, we suggest a lifting-based approach. In the two-party setting, the equality function is a well-known example demonstrating a separation between deterministic and randomized communication complexity. As an application of Theorem~\ref{main_theorem}, we obtain a NOF separation by lifting the equality function..

\begin{theorem}
For any $k \geq 2$ and $r \geq 2^{k+1}$, the deterministic one-way NOF complexity of $\EQUAL \circ \gip_{q,r}^{k}$ is $\Omega(\log q)$, while its randomized one-way NOF complexity is $O(1)$.
\end{theorem}

Note that for any constant $k$, we have $\log q = \Theta(\log q^{r}) = \Theta(\log N)$. Therefore, our construction yields an $\Omega(\log N)$ vs.\ $O(1)$ optimal separation for a constant number of players in the one-way NOF setting.

\paragraph{One-round vs. two-round deterministic separation.} Understanding the power of interaction is one of the central objectives in communication complexity \cite{PS84,DGS87,NW93,VW07,mao2025gadgetless}. As another application of our lifting theorem, we obtain an explicit optimal separation between one-round and two-round deterministic NOF communication. Papadimitriou and Sipser~\cite{PS84} initiated the study of how restricting two-party communication protocols to $r$ rounds affects their complexity. Several researchers subsequently explored this fundamental question. Notably, Duris, Galil, and Schnitger~\cite{DGS87} established an exponential separation between $r$ and $r + 1$ rounds in the two-party setting, a result later improved by Nisan and Wigderson~\cite{NW93}. However, for NOF settings involving more than two parties ($k > 2$), the only known result is by Viola and Wigderson \cite{VW07}, they give $\Omega\left(\frac{n^{1/(k-1)}}{k^k}\right)$ vs $O(\log n)$ separation between one-round and two-round deterministic NOF communication via the tree pointer jumping problem. We improve this bound to $\Omega\left(\frac{n}{2^k}\right)$ vs $O(\log n)$. 

The following \textit{index problem} is a well-known example that separates one-round and two-round deterministic communication in the two-party setting.

\begin{definition}
In the index problem, Alice receives a binary string $x \in \{0,1\}^n$, and Bob receives an index $i \in [n]$. The goal is to compute $x_i$.
\end{definition}

It is well known that the one-round communication complexity of the index problem is $\Omega(n)$, while in the two-round setting, it suffices for Bob to send his input using only $\log n$ bits. We now lift the index problem to the NOF setting. 

To fit our lifting framework, we first slightly modify the input of the index problem. We define the modified index function $\ind : [q] \times [q] \rightarrow \{0,1\}$ as follows: for every input $(x, y)$, let $i_y$ be the integer represented by the first $\log\log q$ bits of $y$, and define $\ind(x, y) = x_{i_y}$. Using previous analysis, we know that the one-round communication complexity of $\ind$ is $\Omega(\log q)$ and the two-round cost is $O(\log\log q)$. 

By applying our lifting theorem (Theorem~\ref{main_theorem}) again, we obtain the following:

\begin{theorem}
For $k \geq 2$ and $r = 2^{k+1}$, the deterministic one-way NOF communication complexity of $\ind \circ \gip_{q,r}^{k}$ is $\Omega(\log q)$, while its two-round NOF communication complexity is $O(\log\log q)$.
\end{theorem}

\paragraph{Disjointness lower bounds.} 
Beyond the lifting theorem, we also apply our technique to obtain a new proof that the deterministic one-way three-party NOF communication cost of set disjointness is at least $\Omega(n)$.

\begin{theorem}\label{DISJ}
The deterministic one-way three-party NOF complexity of set disjointness is $\Omega(n)$.
\end{theorem}

To the best of our knowledge, the only known proof of Theorem~\ref{DISJ} relies on the discrepancy method~\cite{rao2015simplified}. However, the discrepancy method can only establish an \(\Omega(\sqrt{n})\) lower bound in the randomized \textit{one-way NOF} protocols for set disjointness. In contrast, our approach is based on lifting, which has the potential for proving $\Omega(n)$ lower bounds in the randomized setting.

\section{Preliminaries}
\paragraph{Communication complexity.}
We begin by recalling some standard definitions in communication complexity. In the two-party communication model, Alice and Bob receive inputs \( x \in X \) and \( y \in Y \), respectively. Their goal is to compute a function \( f : X \times Y \rightarrow \{0,1\} \). For any two-party function \( f \), we also use \( M(f) \) to denote the \textit{communication matrix} corresponding to \( f \); that is, \( M(f) \) is an \( X \times Y \) matrix where each entry at position \( (x, y) \) is \( f(x, y) \).

\begin{definition}[One-way communication complexity]
 Alice sends a single message \( \Pi(x) \) to Bob, and Bob outputs \( f(x, y) \) based on \( y \) and the received message. The deterministic communication complexity is the maximum length of the message \( |\Pi(x)| \) over all possible inputs, denoted by \( \mathrm{OCC}(f) \).
\end{definition}

We use the following lemma by Feder, Kushilevitz, Naor, and Nisan \cite{feder1995amortized} to characterize the one-way deterministic communication complexity of any two-party function.

\begin{lemma}\label{One_way} 
Let $f : [q] \times [q] \rightarrow \{0,1\}$. Then $f$ has one-way deterministic communication complexity $C$ if and only if $M(f)$ contains a set of $2^C$ distinct rows. That is, there exists a subset $Z \subseteq [q]$ of size $|Z| = 2^C$ such that for any distinct $z_0, z_1 \in Z$, there exists $v \in [q]$ with $f(z_0, v) \neq f(z_1, v)$.
\end{lemma}

\begin{definition}[One-way NOF]
In the $k$-party one-way NOF, $k$ players collaborate to compute a function \( f : X_1 \times \cdots \times X_k \rightarrow \{0,1\} \). The inputs are distributed such that each player $i$ knows all inputs except for $x_i$.

In the one-way communication setting, the players communicate in a fixed order, from the first player to the last. Each player sends a single message, and the last player outputs the value of \( f(x_1, \dots, x_k) \).
\end{definition}

The deterministic communication complexity is defined as the maximum total length of all messages over all possible inputs, and is denoted by \( \mathrm{OCC}^{\mathrm{NOF}}(f) \). The notion of cylinder intersections plays a central role in studying the communication complexity of NOF problems.

\begin{definition}
A set \( S \subseteq X_1 \times \cdots \times X_k \) is called a cylinder if there exists an index \( i \in [k] \) such that membership in \( S \) does not depend on the value of \( x_i \). A set \( S \) is called a cylinder intersection if it can be written as \( S = S_1 \cap \cdots \cap S_k \), where each \( S_i \) is a cylinder.
\end{definition}

\section{A Lifting Theorem for One-Way NOF Model}\label{sec_lifting}
We first recall the definition of $\gip_{q,r}^{k}:(\mathbb{F}_q^{r})^{k}\rightarrow\mathbb{F}_q$ over a large field by,
\[
\gip_{q,r}^{k}(x_1, \dots, x_k) = \sum_{j \in [r]} \prod_{i \in [k]} x_{i,j}.
\]
To simplify notations, we also use $[N]$ to denote $\mathbb{F}_{q}^{r}$. Our main theorem aims to show that
\[
\mathrm{OCC}^{\mathrm{NOF}}(f \circ \gip_{q,r}^{k}) = \Theta(\mathrm{OCC}(f)),
\]
for any two-party communication problem \( f : [q] \times [q] \rightarrow \{0,1\} \) and \( r \geq 2^{k+1} \). The upper bound, i.e., \( \mathrm{OCC}^{\mathrm{NOF}}(f \circ \gip_{q,r}^{k}) = O(\mathrm{OCC}(f)) \), is straightforward. Hence, we focus on the lower bound.

\begin{proof}[Proof of Theorem \ref{main_theorem}]
For any two-party function $f$, by Lemma \ref{One_way}, let $Z \subseteq [q]$ be the set of the distinct rows of $M(f)$ of size $|Z| = 2^{\OCC(f)}$. Our goal is to prove that
\[
\NCC(f \circ \gip) = \Omega\left(\log |Z|\right).
\]
We omit the subscripts of $q, r,k$ here as they are fixed throughout the proof.
Let $\Pi$ be any one-way NOF protocol. We show that if the communication complexity of $\Pi$ is $o(\log |Z|)$, then there exist messages $(\pi_1,\dots,\pi_{k})$, as well as distinct rows $z_0^{*}, z_{1}^{*} \in Z$ and inputs $(x_1^*,\dots, x_k^*) \in [N]^{k}$, such that
\begin{itemize}
    \item The first $k$ players output $(\pi_1,\dots,\pi_{k})$ for both inputs $(x_1^*,\dots, x_k^*,z_0^{*})$ and $(x_1^*,\dots, x_k^*,z_1^{*})$.
    \item $(f \circ \gip)(x_1^*, \dots, x_k^*, z_0^*) \neq (f \circ \gip)(x_1^*, \dots, x_k^*, z_1^*)$.
\end{itemize}
The items above imply that $\Pi$ cannot distinguish inputs $(x_1^*,\dots, x_k^*,z_0^{*})$ and $(x_1^*,\dots, x_k^*,z_1^{*})$. Therefore, it is not a deterministic protocol for computing $f \circ \gip$.

Let \( \Pi \) be any one-way protocol with communication bits less than \( (\log |Z|)/3 \). By the pigeonhole principle, there exists a message tuple \( (\pi_1, \dots, \pi_k) \) such that the following set \( A \) has size at least
\[
|A| \geq \frac{N^k \cdot |Z|}{|Z|^{1/3}} = N^k \cdot |Z|^{2/3} \geq 2N^k \cdot \sqrt{|Z|}.
\]
Here, the set \( A \) is defined as
\[
A = \left\{(x_1,\dots,x_k, z) \in [N]^{k} \times Z \; : \; \text{the first $k$ players output } (\pi_1, \dots, \pi_k) \text{ on input } (x_1,\dots,x_k, z) \right\}.
\]

Now we focus on finding indistinguishable pairs $(x^*, y^*, z_0^*)$ and $(x^*, y^*, z_1^*)$ from the set $A$. The following largeness lemma is a crucial component of our proof.

\begin{lemma}\label{lemma:large_intersection}
Let $A \subseteq [N]^{k} \times Z$ be a set of size $|A| \geq 2N^{k}\cdot \sqrt{|Z|}$. Then for uniformly sampled distinct $z, z' \in Z$, we have that,
\[
\mathbb{E}_{z, z'}\left[ \left| \{(x_1,\dots,x_k) : (x_1,\dots,x_k, z) \in A\} \cap \{(x_1,\dots,x_k) : (x_1,\dots,x_k, z') \in A\} \right| \right] \geq \frac{N^k}{|Z|} \geq \frac{q^{kr}}{q} = q^{kr-1}.
\]    
\end{lemma}

\noindent We postpone to proof of Lemma \ref{lemma:large_intersection} to Section \ref{sec: large_intersection}. Now we fix any distinct $z, z'$ such that
\[
\left| \{(x_1,\dots,x_k) : (x_1,\dots,x_k, z) \in A\} \cap \{(x_1,\dots,x_k) : (x_1,\dots,x_k, z') \in A\} \right|\geq q^{kr-1}\]
as our desired $z_0^*$ and $z_1^*$. Since $z_0^*$ and $z_1^*$ are distinct rows of $Z$, there exists a $v \in [q]$ such that $f(z_0^*, v) \neq f(z_1^*, v)$. Our goal now reduces to finding a tuple $(x_1^*,\dots,x_k^*)$ in the intersection, i.e., 
\[
(x_1^*,\dots, x_k^*, z_0^*), (x_1^*,\dots, x_k^*, z_1^*) \in A
\]
such that $\gip(x_1^*,\dots, x_k^*) = v$.

To prove that, a crucial fact is that for any fixed $z$, the set $\{(x_1,\dots,x_k) : (x_1, \dots,x_k, z) \in A\}$ forms a cylinder intersection. Therefore, the intersection 
\[
\{(x_1,\dots,x_k) : (x_1,\dots,x_k, z_0^*) \in A\} \cap \{(x_1,\dots,x_k) : (x_1,\dots,x_k, z_1^*) \in A\}
\]
is also a cylinder intersection of size at least $q^{kr-1}$. We now invoke the following disperser property of the generalized inner product function over cylinder intersection.

\begin{lemma}
\label{disperser}
For $r\geq 2^{k+1}.$ Let $S\subseteq (\mathbb{F}_{q}^{r})^{k}$ be any cylinder intersection of size $|S|\geq q^{r\cdot k-1}$. Then for every $v\in\mathbb{F}_{q}$,  we have that
\[
\Pr_{(x_1,\dots,x_k) \in S}\left[\gip(x_1,\dots,x_k) = v\right] \geq\frac{1}{q} - q\cdot (2k/q)^{4}
\]
\end{lemma}

\noindent We postpone the proof to Section \ref{sec: disperser}. Assuming Lemma~\ref{disperser}, we can choose the desired $(x_1^*,\dots,x_k^{*})$ such that \(\gip(x_1^*, \dots, x_k^*) = v\). We then conclude the proof.

\end{proof}

\subsection{Proof of Lemma \ref{lemma:large_intersection}}

\label{sec: large_intersection}

Recall that $A \subseteq [N]^k \times Z$ is a set of size $|A| \geq 2N^2\cdot \sqrt{|Z|}$, we aim to prove that
\[
\mathbb{E}_{z, z'}\left[ \left| \{(x_1,\dots,x_k) : (x_1,\dots,x_k, z) \in A\} \cap \{(x_1,\dots,x_k) : (x_1,\dots,x_k, z') \in A\} \right| \right] \geq \frac{N^k}{|Z|}.
\]    

We consider a bipartite graph $G=(L\cup R, E)$ defined by $L=Z$, $R=[N]^k$ and $E=\{(z,(x_1,\dots,x_k)): (x_1,\dots,x_k,z)\in A\}$, and prove the following lemma.
\begin{lemma}\label{graph}
Let \( G = (L\cup R,E) \) be a bipartite graph with $|E|\geq 2\cdot \sqrt{|L|}\cdot |R|$. Then for uniformly chosen distinct $\ell, \ell'\in L$, we have that 
\[
\mathbb{E}\left[ |N(\ell) \cap N(\ell')| \right] \geq \frac{|R|}{|L|}.
\]
where $N(\ell)\subseteq R$ denote the neighborhoods of $\ell$ in $R$.
\end{lemma}
\begin{proof}
Let \( \mathbf{1}(\ell, r) := \mathbf{1}\{(\ell, r) \in E\} \) denote the indicator function for whether the edge \( (\ell, r) \) exists in \( E \). Then we have
\[
\mathbb{E}[|N(\ell) \cap N(\ell')|] = \mathbb{E}\left[\sum_{r \in R} \mathbf{1}(\ell, r) \cdot \mathbf{1}(\ell', r)\right] = \sum_{r \in R} \mathbb{E}\left[\mathbf{1}(\ell, r) \cdot \mathbf{1}(\ell', r)\right].
\]
Let \( \deg(r) := \sum_{\ell \in L} \mathbf{1}(\ell, r) \) be the degree of vertex \( r \in R \). Then,
\[
\sum_{r \in R} \mathbb{E}\left[\mathbf{1}(\ell, r) \cdot \mathbf{1}(\ell', r)\right] 
= \frac{2}{|L|(|L| - 1)} \cdot \sum_{r \in R} \sum_{\substack{\ell, \ell' \in L \\ \ell \neq \ell'}} \mathbf{1}(\ell, r) \cdot \mathbf{1}(\ell', r) 
= \frac{2}{|L|(|L| - 1)} \cdot \sum_{r \in R} \binom{\deg(r)}{2}.
\]
Now observe that
\[
\sum_{r \in R} \binom{\deg(r)}{2} 
\geq \frac{1}{2}\cdot \sum_{r \in R} \left(\deg(r) - 1\right)^2
\geq \frac{1}{2|R|} \left( \sum_{r \in R} (\deg(r) - 1) \right)^2 
= \frac{(|E| - |R|)^2}{2|R|}.
\]
The second inequality follows from the Cauchy-Schwarz inequality; the equality uses the identity \( \sum_{r \in R} \deg(r) = |E| \). Therefore,
\[
\mathbb{E}[|N(\ell) \cap N(\ell')|] 
\geq \frac{2}{|L|(|L| - 1)} \cdot \frac{(|E| - |R|)^2}{2|R|} 
\geq \frac{2}{|L|(|L| - 1)} \cdot \frac{(2\sqrt{|L|} - 1)^2 \cdot |R|^2}{2|R|} 
= \frac{(2\sqrt{|L|} - 1)^2 \cdot |R|}{|L|(|L| - 1)} 
\]
where the second inequality uses the assumption that \( |E| \geq 2\sqrt{|L|} \cdot |R| \). Finally, we conclude the proof since \( (2\sqrt{|L|} - 1)^2 \geq |L| - 1 \).
\end{proof}

\subsection{Proof of Lemma~\ref{disperser}}

\label{sec: disperser}

We prove Lemma~\ref{disperser} in this section. Let $r \geq 2^{k+1}$ and $S\subseteq (\mathbb{F}_{q}^{r})^{k}$ be any cylinder intersection of size $|S|\geq q^{r\cdot k-1}$. We aim to prove that for every $v\in\mathbb{F}_{q}$, 
\[
\Pr_{(x_1,\dots,x_k) \in S}\left[\gip(x_1,\dots,x_k) = v\right] \geq\frac{1}{q} - q\cdot (2k/q)^{4}
\]

\begin{proof}
Let $\mathbf{1}:\mathbb{F}_{q}\rightarrow \{0,1\}$ be the indicator function, i.e., $\mathbf{1}(z)=1$ if $z=0$ and $\mathbf{1}(z)=0$ otherwise. We write its Fourier transform as 
\[
\mathbf{1}(z) = \frac{1}{q} \sum_{\alpha \in \mathbb{F}_q} \chi_{\alpha}(z),
\]
where $\chi_{\alpha}$ are the additive characters of $\mathbb{F}_q$. Then,
\begin{align*}
\Pr_{(x_1,\dots,x_k) \in S}\left[\gip(x_1,\dots,x_k) = v\right]
&= \frac{1}{|S|} \sum_{(x_1,\dots,x_k)\in S} \mathbf{1}\left(\gip(x_1,\dots,x_k) -v\right) \\
&= \frac{1}{|S|} \sum_{(x_1,\dots,x_k)\in S} \frac{1}{q}\cdot \sum_{\alpha \in \mathbb{F}_q} \chi_{\alpha}\left(\gip(x_1,\dots,x_k) - v\right) \\
&= \frac{1}{q} + \frac{1}{q\cdot|S|} \sum_{\alpha \neq 0} \chi_{\alpha}(-v) \sum_{(x_1,\dots,x_k)\in S} \chi_{\alpha}\left(\gip(x_1,\dots,x_k)\right).
\end{align*}
The last equality follows by the fact that $\chi_{\alpha}(0)=1$ and  $\chi_{\alpha}(a+b) = \chi_{\alpha}(a)\cdot \chi_{\alpha}(b)$. We first analyze the upper bound of 
\[
\left|\sum_{(x_1,\dots,x_k)\in S} \chi_{\alpha}\left(\gip(x_1,\dots,x_k)\right)\right|
\]
for every $\alpha$. Recall the definition of $\gip$, we rewrite it as below $D_{\alpha}$.
\[
D_{\alpha}:=\left|\sum_{(x_1,\dots,x_k) \in S}  \chi_{\alpha}\left( \sum_{j\in[r]} \prod_{i\in[k]} x_{i,j} \right) \right| 
= \left|\sum_{(x_1,\dots,x_k) \in (\mathbb{F}_{q}^{r})^{k}} \prod_{i} \mathbf{1}_{S_{i}}(x_1,\dots,x_{k})\cdot  \chi_{\alpha}\left( \sum_{j} \prod_{i} x_{i,j} \right) \right| 
\]
\noindent Let $\gamma_{\alpha}  = D_{\alpha}/q^{r\cdot k} =  \left|\mathbb{E}_{(x_1,\dots,x_k) } \left[\prod_{i} \mathbf{1}_{S_{i}}(x_1,\dots,x_{k})\cdot  \chi_{\alpha}\left( \sum_{j} \prod_{i} x_{i,j} \right) \right]\right| $. By the fact that $(\mathbb{E}[X])^2\leq \mathbb{E}[X^2]$, we have that 

\begin{align*}
\gamma_{\alpha}^2\leq &\mathbb{E}_{x_1,\dots,x_{k-1}} \left[\mathbf{1}_{S_{k}}(x_1,\dots,x_{k-1})^2\mathbb{E}_{x_{k},x_{k}'} \left[ \prod_{i\leq k-1} \mathbf{1}_{S_{i}}(x_1,\dots,x_{k}) \mathbf{1}_{S_{i}}(x_1,\dots,x_{k}')\chi_{\alpha}\left( \sum_{j} x_{k,j} x_{k,j}'\prod_{i\leq k-1} x_{i,j} \right)\right]\right] \\
\leq &\mathbb{E}_{x_{k},x_{k}'} \mathbb{E}_{x_1,\dots,x_{k-1}} \left[ \prod_{i\in[k-1]} \left(\mathbf{1}_{S_{i}}(x_1,\dots,x_{k})\cdot \mathbf{1}_{S_{i}}(x_1,\dots,x_{k}')\right)  \chi_{\alpha}\left( \sum_{j} x_{k,j}\cdot x_{k,j}'\prod_{i\in[k-1]} x_{i,j} \right) \right]
\end{align*}
By applying the this argument $(k-1)$ times, we have that
\[
\gamma_{\alpha}^{2^{k-1}}\leq \mathbb{E}_{(x_2,\dots,x_{k}),(x_2',\dots,x_{k}')} \left[\left|\mathbb{E}_{x_1} \left[ \chi_{\alpha}\left( \sum_{j} x_{1,j}\cdot \prod_{i>1} x_{i,j}x_{i,j}' \right) \right]\right|\right]
\]
Recall that for any character $\chi_{\alpha}$, it holds that $\sum_{z\in\mathbb{F}_{q}}\chi_{\alpha}(z)=0$. Hence, for any $(x_2,\dots,x_{k}),(x_2',\dots,x_{k}')$, if there is a $j\in[r]$ such that $\prod_{i>1} x_{i,j}x_{i,j}'\neq 0$, we have that 
\[
\mathbb{E}_{x_1}\left[\chi_{\alpha}\left( \sum_{j} x_{1,j}\cdot \prod_{i>1} x_{i,j}x_{i,j}' \right)\right] = 0 
\]
Therefore,
\[
\gamma_{\alpha}^{2^{k-1}}\leq \Pr_{(x_2,\dots,x_k,x_2',\dots,x_k')}\left[ \forall j\in[r], \prod_{i>1} x_{i,j}x_{i,j}'=0\right]\leq (2k/q)^{r}
\]
It implies that 
\[
D_{\alpha} = q^{r\cdot k} \cdot \gamma_{\alpha} \leq q^{r\cdot k} \cdot (2k/q)^{r/2^{k-1}}
\]
Recall that $|S| \geq q^{r\cdot k -1}$ and $r\geq 2^{k+1}$, then we have that 
\begin{align*}
\Pr_{(x_1,\dots,x_k) \in S}[\gip(x_1,\dots,x_k) = v]
&= \frac{1}{q} + \frac{1}{q\cdot|S|} \sum_{\alpha \neq 0} \chi_{\alpha}(-v) \sum_{(x_1,\dots,x_k)\in S} \chi_{\alpha}\left(\gip(x_1,\dots,x_k)\right) \\
&\geq \frac{1}{q} -\frac{1}{|S|} \left| \sum_{(x_1,\dots,x_k)\in S} \chi_{\alpha}\left(\gip(x_1,\dots,x_k)\right)\right| \\
&\geq \frac{1}{q} - \frac{q^{r\cdot k} \cdot (2k/q)^{r/2^{k-1}}}{|S|} \geq \frac{1}{q} - q\cdot (2k/q)^{4}
\end{align*}
The first inequality is by the triangle inequality and $|\chi_{\alpha}(-v)| = 1$. The claim then follows.
\end{proof}

A bonus result of Lemma \ref{disperser} is that we construct a $k$-party NOF problem that requires $\Omega(n/2^{k})$ randomized communication cost.
\begin{corollary}
For every $n>0$ and $k\leq \log n - 5\log\log n$. Let $q$ be a $n/2^{k+1}$-bit length prime number, and let $r=2^{k+1}$. Define $G:(\mathbb{F}_q^{r})^{k}\rightarrow\{0,1\}$ be a function defined by 
\[
G(x_1,\dots,x_k) = \gip(x_1,\dots,x_k) \mod 2
\]
Then $G$ is a function with input length $n$, and has a randomized NOF lower bounds $\Omega(n/2^{k+1})$
\end{corollary}

\section{A One-Way NOF Lower Bound of Set Disjointness}

In this section, we present a new proof of the \( \Omega(n) \) deterministic one-way communication lower bound for the three-party NOF model. The proof follows a similar structure to that of Theorem~\ref{main_theorem}. However, the key difference is that we no longer have a large gadget with disperser properties. Instead, we incorporate the density increment argument by Yang and Zhang~\cite{simulationmethod,yang2023lifting}.

We first recall that in the two-party communication, for inputs $x,y$, $\mathrm{DISJ}_2(x, y)=1$ if and only if $x$ and $y$ are disjoint, i.e., 
$
\bigwedge_{i=1}^n (\overline{x}_i \vee \overline{y}_i).
$
\begin{definition}
Let $x, y, z \in \{0,1\}^n$. The three-party set disjointness $\mathrm{DISJ}_3$ is defined as
\[
\mathrm{DISJ}_3(x, y, z) := \mathrm{DISJ}_2\big(z,\, x \land y\big) = \bigwedge_{i=1}^n \left( \overline{z}_i \vee \overline{x}_i \vee \overline{y}_i \right).
\]
\end{definition}

\noindent We now show that the deterministic one-way NOF communication complexity of $\mathrm{DISJ}_3$ is $\Omega(n)$.

\begin{proof}[Proof of Theorem \ref{DISJ}]
For any protocol $\Pi$ with communication complexity $o(n)$, we aim to show that there is a message pair $(\pi^*_A, \pi^*_B)$, a pair of distinct inputs $z_0^*, z_1^* $, and a pair $(x^{*}, y^{*}) $, such that:
\begin{itemize}
\item $\Pi_A(y^{*}, z_{0}^{*}) = \Pi_{A}(y^{*}, z_1^{*}) = \pi^*_A$.
\item $\Pi_B(x^{*}, z_{0}^{*}, \pi^*_A) = \Pi_B(x^{*}, z_1^{*}, \pi^*_A) = \pi^*_B$.
\item $\mathrm{DISJ}_3(z_0^{*},x^{*},y^{*}) \neq \mathrm{DISJ}_3(z_1^{*},x^{*},y^{*}).$
\end{itemize}

Define the set $D_0 := \left\{(x, y) \in \{0,1\}^n \times \{0,1\}^n : \bigwedge_{i=1}^n (\overline{x}_i \vee \overline{y}_i)=1\right\}$, i.e., the set of all pairs of disjoint sets. By an averaging argument, there exists a transcript $(\pi^*_A, \pi^*_B)$ such that the set
\[
A := \{(z, x, y) \in \{0,1\}^n \times \{0,1\}^n \times \{0,1\}^n : \Pi_A(y, z) = \pi^*_A, \ \Pi_B(x, z, \pi^*_A) = \pi^*_B, \ \text{and } (x,y)\in D_0 \}
\]
has size at least
\[
|A| \geq 2^n \cdot |D_0| \cdot 2^{-o(n)}.
\]

\noindent Similar to Lemma~\ref{graph}, we now prove the following largeness lemma.

\begin{lemma}\label{HD}
Let $G = (L, R, E)$ be a bipartite graph with $L = \{0,1\}^n$. Suppose the number of edges satisfies $|E| \geq  2^n \cdot |R| \cdot 2^{-\delta\cdot n}$ for some $\delta<0.1$. Then there exist vertices $\ell, \ell' \in \{0,1\}^n$ with Hamming distance $d_H(\ell, \ell') \geq 0.1n$ such that
\[
|N(\ell) \cap N(\ell')| \geq |R| \cdot 2^{-2\delta n-2},
\]
where $N(\ell) \subseteq R$ denotes the neighborhood of $\ell$ in $R$, and $d_H(\ell, \ell')$ denotes the Hamming distance.
\end{lemma}

\begin{proof}
Since $|E| \geq 2^n \cdot |R| \cdot 2^{-\delta n}$, using similar ideas from Lemma~\ref{graph}, we have that 
\[
\mathbb{E}_{\ell, \ell'}\left[ |N(\ell) \cap N(\ell')| \right] 
= \sum_{r \in R} \mathbb{E}_{\ell, \ell'} \left[ \mathbf{1}(\ell, r) \cdot \mathbf{1}(\ell', r) \right] \geq \frac{2}{|L|(|L| - 1)} \cdot \frac{(|E| - |R|)^2}{2|R|} 
\geq |R| \cdot 2^{-2\delta n - 1},
\]
where $\mathbf{1}(\ell, r)$ denotes the indicator that $(\ell, r) \in E$. We split the expectation based on the Hamming distance between $\ell$ and $\ell'$. Specifically, we write $\mathbb{E}_{\ell, \ell'}\left[ |N(\ell) \cap N(\ell')| \right]$ as 
\[
\mathbb{E} \left[ |N(\ell) \cap N(\ell')| \cdot \mathbf{1}\left(d_H(\ell, \ell') < 0.1n\right) \right] 
+ \mathbb{E} \left[ |N(\ell) \cap N(\ell')| \cdot \mathbf{1}\left(d_H(\ell, \ell') \geq 0.1n\right) \right].
\]

\noindent For the first term, we upper bound it by:
\[
\Pr_{\ell, \ell'}\left[ d_H(\ell, \ell') < 0.1n \right] 
\leq 2^{-n} \cdot \sum_{i=0}^{0.1n} \binom{n}{i} 
\leq 2^{-n} \cdot \left( \frac{e}{0.1} \right)^{0.1n} 
\leq 2^{-n/2}.
\]
Hence, we have
\[
\mathbb{E} \left[ |N(\ell) \cap N(\ell')| \cdot \mathbf{1}\left(d_H(\ell, \ell') < 0.1n\right) \right] 
\leq \mathbb{E} \left[ |R| \cdot \mathbf{1}\left(d_H(\ell, \ell') < 0.1n\right) \right] 
\leq |R| \cdot 2^{-n/2}.
\]

\noindent Therefore, the contribution from the second term is at least
\[
\mathbb{E} \left[ |N(\ell) \cap N(\ell')| \cdot \mathbf{1}\left(d_H(\ell, \ell') \geq 0.1n\right) \right] 
\geq |R| \cdot 2^{-2\delta n - 1} - |R| \cdot 2^{-n/2} 
\geq |R| \cdot 2^{-2\delta n - 2},
\]
where the last inequality holds for small constant $\delta < 0.1$. This completes the proof.

\end{proof}

By applying Lemma~\ref{HD} with $L = \{0,1\}^n$, $R = D_0$, and $E = A$, we obtain a pair of distinct strings $z_0^*, z_1^* \in \{0,1\}^n$ with Hamming distance $d_H(z_0^*, z_1^*) \geq 0.1\cdot n$ such that the following set
\[
R := \left\{ (x, y) : \Pi_A(y, z_0^*) = \pi^*_A, \Pi_B(x, z_0^*, \pi^*_A) = \pi^*_B \right\} \cap \left\{ (x, y) : \Pi_A(y, z_1^*) = \pi^*_A, \Pi_B(x, z_1^*, \pi^*_A) = \pi^*_B \right\}
\]
has large intersection with the set $D_0$, specifically,
\[
|R \cap D_0| \geq |D_0| \cdot 2^{-o(n)}.
\]

\noindent Notice that for fixed $(\pi^*_A, \pi^*_B)$ and $z_0^*,z_1^*$, the set $R$ is a rectangle. Unlike the proof of the lifting theorem (Theorem \ref{main_theorem}, we can no longer apply a disperser property on the set $R\cap D_0$ as there is no large gadget. Instead, we take the approach by \cite{simulationmethod, yang2023lifting}. For each $\ell \in [n]$, define the set
\[
D_\ell := \left\{ (x, y) \in \{0,1\}^n \times \{0,1\}^n : x \land y = e_\ell \right\},
\]
where $e_\ell$ is the unit vector with a $1$ at coordinate $\ell$ and $0$ elsewhere. Yang and Zhang proved the following pseudorandomness lemma.

\begin{lemma}[Theorem 1.1 of \cite{yang2023lifting}]
\label{pseudorandom}
Let $R \subseteq \{0,1\}^n \times \{0,1\}^n$ be any rectangle. If
\[
|R \cap D_0| \geq |D_0| \cdot 2^{-c},
\]
then there exists a set $S \subseteq [n]$ of size at least $n - 2c$ such that $R \cap D_\ell \neq \emptyset$ for every $\ell \in S$.
\end{lemma}

\noindent
We include the proof of Lemma~\ref{pseudorandom} in Section~\ref{sec: pseudorandomness} for completeness.

\medskip

Since $d_H(z_0^*, z_1^*) \geq 0.1\cdot n$, let $T \subseteq [n]$ denote the set of coordinates where $z_0^*$ and $z_1^*$ differ. By Lemma~\ref{pseudorandom}, assuming the communication complexity of $\Pi$ is $o(n)$, there exists a set $S \subseteq [n]$ of size $|S| = n - o(n)$ such that for every $i \in S$, there exists a pair $(x^*, y^*) \in R$ satisfying $x^* \land y^* = e_i$.

Since $|T| \geq 0.1n$, we have $T \cap S \neq \emptyset$. Fix any index $j \in T \cap S$. Then:
\begin{itemize}
    \item We have $(z_0^*)_j \neq (z_1^*)_j$ by definition of $T$,
    \item And there exists $(x^*, y^*) \in R$ such that $x^* \land y^* = e_j$.
\end{itemize}
\noindent Combining the two facts, we obtain two inputs $(x^*, y^*, z_0^*) \in A$ and $(x^*, y^*, z_1^*) \in A$ such that
\[
\mathrm{DISJ}_3(z_0^*, x^*, y^*) = (z_0^*)_j \neq (z_1^*)_j = \mathrm{DISJ}_3(z_1^*, x^*, y^*).
\]

\noindent
This shows that $A$ is not monochromatic, thereby $\Pi$ can not solve $\mathrm{DISJ}_3$ completely.

\end{proof}

\section{Open problems}
A natural open problem is to extend our lifting theorem to the randomized setting. Specifically, we propose the following conjecture:

\begin{conjecture}\label{main_conjecture}
For any partial function $f : [q] \times [q] \to \{0,1\}$, we have
\[
\mathrm{ORCC^{NOF}}(f \circ \gip) = \Theta(\mathrm{ORCC}(f)),
\]
where $\mathrm{ORCC}(f)$ denotes the one-way randomized communication complexity of $f$, and $\mathrm{ORCC^{NOF}}(f \circ \gip)$ denotes the one-way number-on-forehead (NOF) randomized communication complexity of $f \circ \gip$.
\end{conjecture}

Proving this conjecture would imply a separation between quantum and classical communication complexity in the one-way NOF model.

Another important open problem is to extend our lifting theorem (Theorem~\ref{main_theorem}) to general communication protocols beyond the one-way setting. Such a result could provide an optimal separation between randomized and deterministic NOF communication.

\normalem
\bibliographystyle{alpha}
\bibliography{reference.bib}

\appendix
\section{Proof of Lemma \ref{pseudorandom}}

\label{sec: pseudorandomness}
Now we prove Lemma~\ref{pseudorandom} by using the density increment argument by~\cite{simulationmethod,yang2023lifting}.

\begin{theorem}[Restatement of Lemma~\ref{pseudorandom}]
Let \( R \subseteq \{0,1\}^n \times \{0,1\}^n \) be any rectangle. If
\[
|R \cap D_0| \geq |D_0| \cdot 2^{-c},
\]
then there exists a subset \( S \subseteq [n] \) of size at least \( n - 2c \) such that \( R \cap D_i \neq \emptyset \) for every \( i \in S \).
\end{theorem}

\noindent Recall that 
\[
D_0 := \left\{ (x, y) \in \{0,1\}^n \times \{0,1\}^n : x \land y = 0^n \right\}, \quad
D_\ell := \left\{ (x, y) \in \{0,1\}^n \times \{0,1\}^n : x \land y = e_\ell \right\}
\]
where \(e_\ell\) denotes the unit vector with a $1$ in the $\ell$-th coordinate and zeros elsewhere.
For any subset \( I \subseteq [n] \) and any \(\ell \in I\), we define
\[
D_0^I := \left\{ (x, y) \in \{0,1\}^I \times \{0,1\}^I : x \land y = 0^I \right\}, \quad
D_\ell^I := \left\{ (x, y) \in \{0,1\}^I \times \{0,1\}^I : x \land y = e_\ell \right\}.
\]

\noindent We now introduce the density function, which will be used in the density increment argument.

\begin{definition}[Density function]
Let \( I \subseteq [n] \), and let \( R = X \times Y \subseteq \{0,1\}^I \times \{0,1\}^I \) be a rectangle. The density of \( R \) with respect to \( D_0^I \) is defined as
\[
E^I(R) := \log \left( \frac{|R \cap D_0^I|}{|D_0^I|} \right).
\]
\end{definition}

\noindent
Note that \( E^I(R) \leq 0 \) for any rectangle \( R \), and that \( |D_0^I| = 3^{|I|} \), as each coordinate allows three disjoint assignments: \( (0,0), (0,1), (1,0) \). We will write \( E(R) \) when the underlying index set \( I \) is clear from context.

A crucial step in our argument is the projection operation, which allows us to restrict attention to a subset of coordinates while preserving the rectangle structure.

\begin{definition}[Projection]
Let \( R = X \times Y \subseteq \{0,1\}^I \times \{0,1\}^I \) be a rectangle.  
For any coordinate \( i \in I \) and side \( C \in \{X, Y\} \), the projection of \( R \) onto coordinate \( i \) from side \( C \) is defined as a rectangle \( \Pi_{i,C}(R) = X' \times Y' \subseteq \{0,1\}^{I \setminus \{i\}} \times \{0,1\}^{I \setminus \{i\}} \), where:
\begin{itemize}
    \item If \( C = X \), then $X' = \left\{ x_{I \setminus \{i\}} : x \in X,\, x_i = 0 \right\}$ and $Y' = \left\{ y_{I \setminus \{i\}} : y \in Y \right\}$.
    
    \item If \( C = Y \), then $X' = \left\{ x_{I \setminus \{i\}} : x \in X \right\}$ and $
    Y' = \left\{ y_{I \setminus \{i\}} : y \in Y,\, y_i = 0 \right\}$.
\end{itemize}
Here, \( x_{I \setminus \{i\}} \) denotes the restriction of the string \( x \in \{0,1\}^I \) to the coordinates in \( I \setminus \{i\} \).
\end{definition}

The projection operation satisfies two useful properties. The first property is that projection preserves the absence of specific structured sets:

\begin{fact}\label{fact: k_u_m}
Let \( R \subseteq \{0,1\}^I \times \{0,1\}^I \) be a rectangle such that \( R \cap D_j^I = \emptyset \) for some \( j \in I \). Then, for every \( i \in I \) and every \( C \in \{X, Y\} \), we have
\[
\Pi_{i,C}(R) \cap D_j^{I \setminus \{i\}} = \emptyset.
\]
\end{fact}

\noindent
The proof of Fact~\ref{fact: k_u_m} follows directly from the definition of projection and is omitted.

The second property quantifies how projection can increase the density, and is captured by the following projection lemma:
\begin{lemma}[Projection Lemma]\label{lemma: k_u_p}
Let \( R = X \times Y \subseteq \{0,1\}^I \times \{0,1\}^I \) be a rectangle. If there exists a coordinate \( i \in I \) such that \( R \cap D_i^I = \emptyset \), then there exists \( C \in \{X, Y\} \) such that
\[
E^{I \setminus \{i\}} \big( \Pi_{i,C}(R) \big) \geq E^I(R) + \frac{1}{2}.
\]
\end{lemma}

Given Lemma~\ref{lemma: k_u_p} and Fact~\ref{fact: k_u_m}, the proof of Lemma~\ref{pseudorandom} becomes straightforward. We iteratively apply the projection operation on coordinates \( i \notin S \), choosing at each step a suitable side \( C \in \{X, Y\} \) as guaranteed by Lemma~\ref{lemma: k_u_p}, which increases the density function by at least \( 1/2 \) per step. Since the density is upper bounded by zero, this process can be repeated at most \( 2c \) times, yielding a subset \( S \subseteq [n] \) of size at least \( n - 2c \) such that \( R \cap D_i \neq \emptyset \) for all \( i \in S \).

We now proceed to prove Lemma~\ref{lemma: k_u_p}.

\begin{proof}[Proof of Lemma~\ref{lemma: k_u_p}]
Let \( R \subseteq \{0,1\}^I \times \{0,1\}^I \) be a rectangle such that \( R \cap D_i^I = \emptyset \). Define \( I' := I \setminus \{i\} \), and let
\[
L := \left\{ (x', y') \in D_0^{I'} : \exists (x, y) \in R \cap D_0^I \text{ such that } x_{I'} = x',~ y_{I'} = y' \right\}.
\]
Note that for any \( C \in \{X, Y\} \), we have
\[
\Pi_{i,C}(R) \cap D_0^{I'} = \Pi_{i,C}(R) \cap L.
\]
Our goal is to show that there exists \( C \in \{X,Y\} \) such that \( |\Pi_{i,C}(R) \cap L| \) is large.

For each \( (x', y') \in L \), define the extension set
\[
\ext(x', y') := \left\{ (x, y) \in R \cap D_0^I : x_{I'} = x',~ y_{I'} = y' \right\}.
\]
We now show that
\begin{equation} \label{eq:ext_k}
|\ext(x', y')| \leq 2.
\end{equation}

\noindent This follows from the assumption \( R \cap D_i^I = \emptyset \). Suppose, for contradiction, that \( |\ext(x', y')| = 3 \). Then all three extensions of \( (x', y') \) to coordinate \( i \) namely, \( (x_i, y_i) = (0,0), (1,0), (0,1) \) must be present in \( R \). by the rectangle property of $R$, which would imply exists a pair \( (x, y) \in R \cap D_i^I \) with $(x_i, y_i) = (1,1)$, a contradiction. Hence, \eqref{eq:ext_k} holds.

Next, we partition \( L \) into two sets:
\[
A := \left\{ (x', y') \in L : |\ext(x', y')| = 2 \right\}, \quad
B := \left\{ (x', y') \in L : |\ext(x', y')| = 1 \right\}.
\]
Observe that for any \( (x', y') \in A \), both \( (x',0) \in X \) and \( (y',0) \in Y \), since \( R \) is a rectangle and both extensions with \( x_i = 0 \) and \( y_i = 0 \) must be in \( R \). Therefore,
\[
(x', y') \in \Pi_{i,C}(R) \quad \text{for all } C \in \{X,Y\},
\]
which implies $|A| = |A \cap \Pi_{i,C}(R)| \quad \text{for any } C$.

For any such \( C \), by \eqref{eq:ext_k}, we have
\[
2 \cdot |A \cap \Pi_{i,C}(R)| = 2 \cdot |A| \geq \left| \left\{ (x, y) \in R \cap D_0^I : (x_{I'}, y_{I'}) \in A \right\} \right|.
\]

For the \( B \) part, note that each \( (x', y') \in B \) corresponds to a unique element in \( R \cap D_0^I \), so
\[
\left| \left\{ (x, y) \in R \cap D_0^I : (x_{I'}, y_{I'}) \in B \right\} \right| = |B|.
\]
Moreover, for each \( (x', y') \in B \), there exists at least one choice of \( C \in \{X,Y\} \) such that \( (x', y') \in \Pi_{i,C}(R) \). By an averaging argument, there exists \( C \in \{X, Y\} \) such that
\[
2 \cdot |B \cap \Pi_{i,C}(R)| \geq |B| = \left| \left\{ (x, y) \in R \cap D_0^I : (x_{I'}, y_{I'}) \in B \right\} \right|.
\]

Putting both parts together, we have for this fixed \( C \),
\begin{align*}
2 \cdot |L \cap \Pi_{i,C}(R)|
&= 2 \cdot |A \cap \Pi_{i,C}(R)| + 2 \cdot |B \cap \Pi_{i,C}(R)| \\
&\geq \left| \left\{ (x, y) \in R \cap D_0^I : (x_{I'}, y_{I'}) \in A \right\} \right| 
+ \left| \left\{ (x, y) \in R \cap D_0^I : (x_{I'}, y_{I'}) \in B \right\} \right| \\
&= |R \cap D_0^I|.
\end{align*}

Finally, using the definition of the density function:
\begin{align*}
E^{I'}(\Pi_{i,C}(R))
&= \log \left( \frac{|\Pi_{i,C}(R) \cap D_0^{I'}|}{3^{|I'|}} \right)
= \log \left( \frac{3 \cdot |\Pi_{i,C}(R) \cap L|}{3^{|I|}} \right) \\
&\geq \log \left( \frac{3 \cdot |R \cap D_0^I|}{2 \cdot 3^{|I|}} \right)
= E^I(R) + \log(3/2) \geq E^I(R) + \frac{1}{2}.
\end{align*}
This completes the proof.
\end{proof}

\end{document}